\documentclass[12pt]{article}
\usepackage{amsmath,mathrsfs}
\usepackage{amssymb}
\usepackage{amsthm}
\usepackage{indentfirst}
\allowdisplaybreaks

\newtheorem*{remark}{Remark}
\newtheorem{theorem}{Theorem}
\newtheorem*{theorem*}{Theorem}

\newtheorem{proposition}{Proposition}
\newtheorem*{proposition*}{Proposition}
\newcommand{\R}{\mathbb{R}}
\newcommand{\Rm}{\mathbb{R}^m}

\newcommand{\Dtwo}{\mathcal{D}_k}

\newcommand{\bx}{\boldsymbol{x}}
\newcommand{\ba}{\boldsymbol{a}}

\newcommand{\bu}{\boldsymbol{u}}

\newcommand{\be}{\begin{eqnarray*}}
\newcommand{\ee}{\end{eqnarray*}}
\newcommand{\bs}{\boldsymbol}

\newcommand{\Sm}{\mathbb{S}^{m-1}}

\newcommand{\Mk}{\mathcal{M}_k}

\newcommand{\Mkk}{\mathcal{M}_{k-1}}
\newcommand{\HK}{\mathcal{H}_k}
\newcommand{\Hk}{\mathcal{H}_k}

\newcommand{\udx}{\langle u,D_x\rangle}
\newcommand{\dudx}{\langle D_u,D_x\rangle}

\newcommand{\bpm}{\begin{pmatrix}}
\newcommand{\epm}{\end{pmatrix}}
\newcommand{\baa}{\begin{align*}}

\begin{document}

\title{Bosonic Laplacians in higher spin Clifford analysis}

\author{Chao Ding$^1$\thanks{Electronic address:  {\tt cding@ahu.edu.cn}.} and John Ryan$^2$\thanks{Electronic address: {\tt jryan@uark.edu.}} \\
\emph{\small $^1$Center for Pure Mathematics, School of Mathematical Sciences,}\\
\emph{\small Anhui University, Hefei, P.R. China}\\
\emph{\small $^2$Department of Mathematical Sciences, University of Arkansas, Fayetteville, AR, USA} }
\date{}

\maketitle

\begin{abstract}
In this article, we firstly introduce higher spin Clifford analysis, which are considered as generalizations of classical Clifford analysis by considering functions taking values in irreducible representations of the spin group. Then, we introduce a type of second order conformally invariant differential operators, named as bosonic Laplacians, in the higher spin Clifford analysis. In particular, we will show their close connections to classical Maxwell equations.  At the end, we will introduce a new perspective to define bosonic Laplacians, which simplifies the connection between bosonic Laplacians and Rarita-Schwinger type operators obtained before. Moreover, a matrix type Rarita-Schwinger operator is obtained and some results related to this new first order matrix type operator are provided.
\end{abstract}
{\bf Keywords:}\quad Bosonic Laplacians,  Matrix type Rarita-Schwinger opreator\\
{\bf AMS subject classification (2020):}\quad 30G35, 42B37
\section{Introduction}
Classical Clifford analysis started as a generalization of aspects of one variable  complex analysis to higher dimensional Euclidean spaces. At the center of this theory is the study of the Dirac operator on $\mathbb{R}^m$. The Dirac operator $D_x$ is considered as a generalization of the role of the Cauchy-Riemann operator. Moreover, this operator is related to Laplacian with $\Delta_x=-D_x^2$. More details on classical Clifford analysis can be found in \cite{Brackx,Delanghe}.
\par
In the past few decades, more researchers have been working on generalizations of classical Clifford analysis to the so called higher spin Clifford analysis. In this context, we investigate higher spin differential operators acting on functions taking values in arbitrary irreducible representations of the spin group. A well-known higher spin differential operator is the Rarita-Schwinger operator, which is a first order conformally invariant differential operator. This differential operator was firstly introduced by Rarita and Schwinger in 1941 to describe the relativistic field equation of spin-$3/2$ fermions. Bure\v s et al. \cite{Bures} firstly generalized the Rarita-Schwinger operator to higher spin cases in the framework of Clifford analysis. Dunkl et al. \cite{Dunkl} revisited the Rarita-Schwinger operator with some analytic approaches in 2014. In \cite{DR1}, the authors showed that the Rarita-Schwinger operator could be constructed as a Stein-Weiss gradient \cite{Stein}.
\par
In 2016, Eelbode et al. \cite{Eelbode} pointed out that $\Lambda_x$ is no longer conformally invariant in the higher spin spaces. Instead, they introduced a second order conformally invariant differential operator, which is named as the generalized Maxwell operator. It turns out that the generalized Maxwell operator reduces to the source-free classical Maxwell equations with Lorentz conditions in the Minkowski space. In 2017, De Bie et al. \cite{DeBie} generalized the generalized Maxwell operator to the so called the higher spin Laplace operators (also called bosonic Laplacians).  In  2019, Ding and Ryan discovered a Borel-Pompeiu formula and a Green type integral formula for bosonic Laplacians. Further, Ding et al. investigated some boundary value problems for bosonic Laplacians in the upper half-space and the unit ball. More importantly, with the uniqueness of the solution to the Dirichlet problem of bosonic Laplacians, a Poisson integral formula for null solutions to bosonic Laplacians was provided in \cite{DNR} and some other important properties, such as mean value property, Cauchy's estimates, etc. were introduced there.
\par
In this article, we will review the construction, connection to classical Maxwell equations and some important properties for bosonic Laplacians. Some recent progress will also be introduced at the end. This article is organized as follows. In Section 2, we introduced some preliminaries of Clifford algebra. Section 3 is devoted to the motivation of constructing bosonic Laplacians, which is their connections to classical Maxwell equations. In Section 4, we reviewed Rarita-Schwinger type operators and the decomposition of bosonic Laplacians in terms of Rarita-Schwinger type operators. In Section 5, we introduced a matrix form for bosonic Laplacians, which leads to a compact form of Borel-Pompeiu integral formula.
\section*{Acknowledgement}
This paper is dedicated to the memory of Wolfgang Tutschke. The authors are also grateful to the referees for helpful comments.

\section{Preliminaries}
Let $\{\bs{e}_1,\cdots,\bs{e}_m\}$ be a standard orthonormal basis for the $m$-dimensional Euclidean space $\R^m$. Suppose $\bx$ and $\ba$ are two arbitrary vectors in $\R^m$, we now show that a reflection of $\bx$ across the hyperplane perpendicular to $\ba$ can simply be written in the form $\ba\bx\ba$. This simplifies the calculation significantly later in this article. To explain this we need to introduce Clifford algebras.
\par
The (real) Clifford algebra $\mathcal{C}l_m$ is generated by $\R^m$ with the relationship $\bs{e}_i\bs{e}_j+\bs{e}_j\bs{e}_i=-2\delta_{ij},\ 1\leq i,j\leq m.$
An arbitrary element of the basis of the Clifford algebra can be written as $\bs{e}_A=\bs{e}_{j_1}\cdots \bs{e}_{j_r},$ where $A=\{j_1, \cdots, j_r\}\subset \{1, 2, \cdots, m\}$ and $1\leq j_1< j_2 < \cdots < j_r \leq m$. The $m$-dimensional Euclidean space $\R^m$ can be embedded into $\mathcal{C}l_m$ with the mapping
$\bx=(x_1,\cdots,x_m)\ \mapsto\quad \sum_{j=1}^mx_j\bs{e}_j.$
For $\bx\in\R^m$, one can easily see that $|\bx|^2=\sum_{j=1}^mx_j^2=-\bx^2$.
\par
Suppose that $\ba\in \mathbb{S}^{m-1}\subseteq \mathbb{R}^m$, if we consider $\ba\bx\ba$, we may decompose
$$\bx=\bx_{\ba\parallel}+\bx_{\ba\perp},$$
where $\bx_{\ba\parallel}$ is the projection of $\bx$ onto $\ba$ and $\bx_{\ba\perp}$ is the rest, perpendicular to $a$. Hence $\bx_{\ba\parallel}$ is a scalar multiple of $\ba$ and we have
$$\ba\bx\ba=\ba\bx_{\ba\parallel}\ba+\ba\bx_{\ba\perp}\ba=-\bx_{\ba\parallel}+\bx_{\ba\perp}.$$
So the action $\ba\bx\ba$ describes a reflection of $\bx$ across the hyperplane perpendicular to $\ba$.
\par
Let $\HK$ ($1 \leq k \in {\mathbb N}$) be the space of real-valued homogeneous harmonic polynomials of degree $k$ in $m$-dimensional Euclidean space. If we consider a function $f(\bx,\bu)\in C^{\infty}(\R^m\times\R^m,\HK)$, i.e., for a fixed $\bx\in\R^m$, $f(\bx,\bu)\in\HK$ with respect to $\bu \in \R^m$. Recall that bosonic Laplacians \cite{Eelbode} are defined by
\begin{eqnarray}\label{Dtwo}
&&\Dtwo:\ C^{\infty}(\R^m\times\R^m,\HK)\longrightarrow C^{\infty}(\R^m\times\R^m,\HK),\nonumber\\
&&\Dtwo=\Delta_{\bx}-\frac{4\langle \bu,\nabla_{\bx}\rangle\langle \nabla_{\bu},\nabla_{\bx}\rangle}{m+2k-2}+\frac{4|\bu|^2\langle \nabla_{\bu},\nabla_{\bx}\rangle^2}{(m+2k-2)(m+2k-4)},
\end{eqnarray}
where $\langle\ ,\ \rangle$ is the standard inner product in $\R^m$, $\nabla_{\bx}$ is the gradient with respect to $\bx$. In particular, $
\mathcal{D}_1=\Delta_{\bx}-\frac{4}{m}\langle \bu,\nabla_{\bx}\rangle\langle \nabla_{\bu},\nabla_{\bx}\rangle
$
 is the generalized Maxwell operator introduced in \cite{Eelbode}.
\section{Connection to source-free classical Maxwell equations}
In physics, the classical source-free coupled Maxwell equations are given as the following
\be
&&\nabla\cdot\bs{E}=0,\ \nabla\times\bs{E}=-\frac{\partial{\bs{B}}}{\partial t},\
\nabla\cdot\bs{B}=0,\ \nabla\times\bs{B}=\mu_0\epsilon_0\frac{\partial\bs{E}}{\partial t},
\ee
Where $\bs{E}$ is the electric field, $\bs{B}$ is the magnetic field, $\bs{B}$ and $\bs{E}$ are both vector fields in $\mathbb{R}^3$. $\mu_0$ is the permeability of free space and $\epsilon_0$ is the permittivity of free space.
\par
With the Lorenz condition, there is a different formulation of Maxwell equations in terms of a differential 2-form $F^{\alpha\beta}$, known as the Maxwell-Faraday tensor, $\partial_{\alpha}F^{\alpha\beta}=0$, which can also be rewritten as
$
\square C^{\beta}-\partial^{\beta}\partial_{\alpha}C^{\alpha}=0,
$
where $C^{\alpha}=(\Phi,\bs{C})$ is usually called an electromagenetic four-potential, $\partial^{\alpha}=(\partial_t,-\nabla)$ and $\partial_{\alpha}=(\partial_t,\nabla)$, more details can be found in \cite{Jackson}.
\par
The equation $\square C^{\beta}-\partial_{\beta}\partial^{\alpha}C^{\alpha}=0$ in the Minkowski space has a generalization to the $m$-dimensional Euclidean space given below.
\be
\Delta_{\bs{x}}\bs{f_s}(\bs{x})-\frac{4}{m}\sum_{j=1}^{m}\partial_{x_s}\partial_{x_j}\bs{f_j}(\bs{x})=0,\ 1\leq s\leq m.
\ee
where $f_s(\bx)$ is a vector valued function, with $\bx\in\R^m$. The constant $4/m$ allows the generalized operator above to preserve the conformal invariance property of the Maxwell equations. Notice that the space of real-valued homogeneous harmonic polynomials with degree-$1$ with respect to a variable $\bu\in\R^m$ , denoted by $\mathcal{H}_1(\R)$ , is spanned by $\{u_1,\cdots,u_m\}$. The $m$ equations above can be replaced by one equation
\be
\sum_{s=1}^{m}u_s\bigg(\Delta_{\bs{x}}\bs{f_s}(\bs{x})-\frac{4}{m}\sum_{j=1}^{m}\partial_{x_s}\partial_{x_j}\bs{f_j}(\bs{x})\bigg)=0.
\ee
This equation can also be rewritten as
\be
0&=&\sum_{s=1}^m\Delta_{\bx}u_sf_s(\bx)-\frac{4}{m}\sum_{j,s=1}^mu_s\partial_{x_s}\partial_{x_j}f_j(\bx)\\
&=&\sum_{s=1}^m\Delta_{\bx}u_sf_s(\bx)-\frac{4}{m}\sum_{j,s,k=1}^mu_s\partial_{x_s}\partial_{x_j}\partial_{u_j}u_kf_k(\bx).
\ee
Now, we consider $f(\bx,\bu):=\sum_{s=1}^mu_sf_s(\bx)$ as a $\mathcal{H}_1(\R)$-valued function on $\R^m$. With the help of the Dirac operator $\bs{D}_{\bx}=\sum_{s=1}^m\bs{e}_s\partial_{x_s}$ in the $m$-dimensional Euclidean space, this equation can be written in a compact form
\be
\left(\Delta_{\bs{x}}-\frac{4}{m}\langle \bs{u},\bs{D}_{\bs{x}}\rangle\langle \bs{D}_{\bs{u}},\bs{D}_{\bs{x}}\rangle\right)\bs{f}(\bs{x},\bs{u})=0,
\ee
and the operator $\Delta_{\bs{x}}-\frac{4}{m}\langle \bs{u},\bs{D}_{\bs{x}}\rangle\langle \bs{D}_{\bs{u}},\bs{D}_{\bs{x}}\rangle$, denoted by $\mathcal{D}_1$ is the generalized Maxwell operator \cite{Eelbode}.
\par
More generally, if we consider a function $f(\bx,\bu)\in C^{\infty}(\R^m,\HK(\R))$, i.e., for a fixed $\bx\in\R^m$, $f(\bx,\bu)\in\HK(\R)$ with respect to $\bu$, where $\HK(\R)$ stands for the space of real-valued homogeneous harmonic polynomials of degree $k$. The second order conformally invariant differential operators, named as bosonic Laplacians (also known as the higher spin Laplace operators \cite{Eelbode}), are defined as
\begin{eqnarray*}
&&\Dtwo:\ C^{\infty}(\R^m,\HK(\R))\longrightarrow C^{\infty}(\R^m,\HK(\R)),\nonumber\\
&&\Dtwo=\Delta_{\bx}-\frac{4\langle \bu,D_{\bx}\rangle\langle D_{\bu},D_{\bx}\rangle}{m+2k-2}+\frac{4|u|^2\langle D_{\bu},D_{\bx}\rangle^2}{(m+2k-2)(m+2k-4)}.
\end{eqnarray*}
This coincides the expression given in \eqref{Dtwo}.
\section{Rarita-Schwinger type operators}

Let $\mathcal{M}_k$ denote the space of $\mathcal{C}l_m$-valued monogenic polynomials homogeneous of degree $k$. Note that if $h_k(u)\in\Hk$, the space of $\mathcal{C}l_m$-valued harmonic polynomials homogeneous of degree $k$, then $D_uh_k(u)\in\mathcal{M}_{k-1}$, but $D_uup_{k-1}(u)=(-m-2k+2)p_{k-1}(u),$ where $p_{k-1}(u)\in \Mkk$. Hence,
\begin{eqnarray}\label{Almansi}
\mathcal{H}_k=\mathcal{M}_k\oplus u\mathcal{M}_{k-1},\ h_k=p_k+up_{k-1}.
\end{eqnarray}
This is an \emph{Almansi-Fischer decomposition} of $\Hk$ \cite{Dunkl}. In this Almansi-Fischer decomposition, we have $P_k^+$ and $P_k^-$ as the projection maps
\begin{eqnarray}
&&P_k^+=I+\frac{uD_u}{m+2k-2}:\ \mathcal{H}_k\longrightarrow \mathcal{M}_k, \label{Pk+}\\
&&P_k^-=I-P_k^+=\frac{-uD_u}{m+2k-2}:\ \mathcal{H}_k\longrightarrow u\mathcal{M}_{k-1} \label{Pk-}.
\end{eqnarray}
Suppose $U$ is a domain in $\mathbb{R}^m$. Consider a differentiable function $f: U\times \mathbb{R}^m\longrightarrow \mathcal{C}l_m$
such that, for each $x\in U$, $f(x,u)$ is a left monogenic polynomial homogeneous of degree $k$ in $u$. Then \textbf{the Rarita-Schwinger operator} \cite{Bures,Dunkl} is defined by
 $$R_k=P_k^+D_x:\ C^{\infty}(\Rm,\Mk)\longrightarrow C^{\infty}(\Rm,\Mk).$$
 We also need the following three more Rarita-Schwinger type operators.
 \begin{eqnarray*}
	&&\text{\textbf{The twistor operator:}}\ T_k=P_k^+D_x:\ C^{\infty}(\Rm,u\Mkk)\longrightarrow C^{\infty}(\Rm,\Mk),\\
	&&\text{\textbf{The dual twistor operator:}}\ T_k^*=P_k^-D_x:\ C^{\infty}(\Rm,\Mk)\longrightarrow C^{\infty}(\Rm,u\Mkk),\\
	&&\text{\textbf{The remaining operator:}}\ Q_k=P_k^-D_x:\ C^{\infty}(\Rm,u\Mkk)\longrightarrow C^{\infty}(\Rm,u\Mkk).
	\end{eqnarray*}
	More details can be found in \cite{Bures,Dunkl}. Let $Z_k^1(u,v)$ be the reproducing kernel for $\Mk$, which satisfies
\begin{eqnarray*}
f(v)=\int\displaylimits_{\Sm}\overline{Z_k^1(u,v)}f(u)dS(u),\ for\ all\ f(v)\in\Mk.
\end{eqnarray*}
Then the fundamental solution for $R_k$ (\cite{Dunkl}) is
\begin{eqnarray*}
E_k(x,u,v)=\frac{1}{\omega_{m}a_k}\frac{x}{||x||^m}Z_k^1(\frac{xux}{||x||^2},v),
\end{eqnarray*}
where the constant $a_k$ is $\displaystyle\frac{m-2}{m+2k-2}$ and $\omega_{m}$ is the area of the $m$-dimensional unit sphere. Similarly, we have the fundamental solution for $Q_k$ (\cite{Li}) as follows
\begin{eqnarray*}
F_k(x,u,v)=\frac{-1}{\omega_{m}a_k}u\frac{x}{||x||^m}Z_{k-1}^1(\frac{xux}{||x||^2},v)v.
\end{eqnarray*}
The bosonic Laplacians are constructed in \cite{DeBie}, and is defined as follows.
\begin{eqnarray*}
\Dtwo=\Delta_x-\displaystyle\frac{4\udx\dudx}{m+2k-2}+\displaystyle\frac{4||u||^2\dudx^2}{(m+2k-2)(m+2k-4)}.
\end{eqnarray*}
 The fundamental solution for $\Dtwo$ is also provided in the same reference. We denote it by
\begin{eqnarray*}
H_k(x,u,v)=\displaystyle\frac{(m+2k-4)\Gamma(\displaystyle\frac{m}{2}-1)}{4(4-m)\pi^{\frac{m}{2}}}||x||^{2-m}Z_k^2(\displaystyle\frac{xux}{||x||^2},v),
\end{eqnarray*}
where $Z_k^2(u,v)$ is the reproducing kernel for $\Hk$ and satisfies
\begin{eqnarray*}
f(v)=\int\displaylimits_{\Sm}\overline{Z_k^2(u,v)}f(u)dS(u),\ for\ all\ f(v)\in\Hk.
\end{eqnarray*}
\begin{theorem}\label{thmone}
Let $P_k^+$ and $P_k^-$ be the projection maps defined in (\ref{Pk+}) and (\ref{Pk-}). Then the higher spin Laplace operator $\Dtwo$ can be written as follows.
\begin{eqnarray}
\Dtwo&=&-R_k^2P_k^++\frac{2T_k^*R_kP_k^+}{m+2k-4}-\frac{2T_kQ_kP_k^-}{m+2k-4}-\frac{(m+2k)Q_k^2P_k^-}{m+2k-4}\label{D21}\\
         &=&-R_k^2P_k^++\frac{2R_kT_kP_k^-}{m+2k-4}-\frac{2Q_kT_k^*P_k^+}{m+2k-4}-\frac{(m+2k)Q_k^2P_k^-}{m+2k-4},\nonumber
\end{eqnarray}
when it acts on a function $f(x,u)\in C^{\infty}(\Rm,\Hk).$
\end{theorem}
\section{A matrix form for bosonic Laplacians}
We notice that the connection between bosonic Laplacians and Rarita-Schwinger type operators given in Theorem \ref{thmone} leads to a complicated  form of Borel-Pompeiu integral formula in \cite{DR}. Hence, in this section, we introduce a matrix form of bosonic Laplacians, which simplifies the formulation of the equations given in Theorem \ref{thmone}. Further, the matrix form also gives a compact form of Borel-Pompeiu integral formula, which is more convenient to use for our later work on boundary value problems related to bosonic Laplacians.
\par
Using Almansi-Fischer decomposition $\Hk=\Mk\oplus u\Mkk$, given $f(x,u)\in C^{\infty}(\Rm,\Hk)$, one can decompose it as $f(x,u)=f_1(x,u)+f_2(x,u)$, where $f_1\in C^{\infty}(\Rm,\Mk)$ and $f_2\in C^{\infty}(\Rm,u\Mkk)$. Therefore, one can consider $f(x,u)\in C^{\infty}(\Rm,\Hk)$ as a vector $\begin{pmatrix} f_1\\f_2\end{pmatrix}:=\bs{f}$. Further, we notice that $P_k^+f_2=P_k^-f_1=0$, which says that
\begin{align*}
\mathcal{D}_kf=&\bigg(-R_k^2P_k^++\frac{2R_kT_kP_k^-}{m+2k-4}-\frac{2Q_kT_k^*P_k^+}{m+2k-4}-\frac{(m+2k)Q_k^2P_k^-}{m+2k-4}\bigg)(f_1+f_2)\\
=&-R_k^2P_k^+f_1+\frac{2R_kT_kP_k^-}{m+2k-4}f_2-\frac{2Q_kT_k^*P_k^+}{m+2k-4}f_1-\frac{(m+2k)Q_k^2P_k^-}{m+2k-4}f_2,
\end{align*}
where $f,\ f_1$ and $f_2$ are defined as above. With the vector form of a function $f$, this motivates us to rewrite $\Dtwo$ as the following
\be
\bs{D}_2:=\bpm
R_k^2P_k^+ & -\frac{2R_kT_kP_k^-}{m+2k-4}\\
\frac{2Q_kT_k^*P_k^+}{m+2k-4} & \frac{m+2k}{m+2k-4}Q_k^2P_k^-
\epm,
\ee
and the calculation above easily shows that $\mathcal{D}_kf=\bs{D}_2\bs{f}$.
\par
Now, let $\bs{D}_1=\bpm -R_kP_k^+ & \frac{2T_k^*P_k^+}{\sqrt{(m+2k)(m+2k-4)}}\\ \frac{2T_kP_k^-}{m+2k-4} & \sqrt{\frac{m+2k}{m+2k-4}}Q_kP_k^-\epm$, we have a factorization of $\bs{D}_2$ as follows.
\begin{proposition}
We have a factorization of $\bs{D}_2$ given by
$
\bs{D}_2=\bs{D}_1\bs{D}_1^T.
$
\end{proposition}
\begin{proof}
To verify the result above, all we need is to notice that
\begin{align*}
&R_kP_k^+f_1,\ T_kP_k^-f_2\in C^{\infty}(\Rm,\Mk),\quad  T_k^*P_k^+f_1,\ Q_kP_k^-f_2\in  C^{\infty}(\Rm,u\Mkk),\\
&P_k^+T_k^*=P_k^+Q_k=0,\quad P_k^-R_k=P_k^-T_k=0,
\end{align*}
which come from the definitions of the two projections $P_k^+,\ P_k^-$ given in \eqref{Pk+} and \eqref{Pk-}. Then, one can easily see that 
\be
&&\bpm -R_kP_k^+ & \frac{2T_k^*P_k^+}{\sqrt{(m+2k)(m+2k-4)}}\\ \frac{2T_kP_k^-}{m+2k-4} & \sqrt{\frac{m+2k}{m+2k-4}}Q_kP_k^-\epm \bpm -R_kP_k^+ & \frac{2T_kP_k^-}{m+2k-4}\\ \frac{2T_k^*P_k^+}{\sqrt{(m+2k)(m+2k-4)}} & \sqrt{\frac{m+2k}{m+2k-4}}Q_kP_k^-\epm\\
=&&\bpm
R_k^2P_k^+ & -\frac{2R_kT_kP_k^-}{m+2k-4}\\
\frac{2Q_kT_k^*P_k^+}{m+2k-4} & \frac{m+2k}{m+2k-4}Q_k^2P_k^-
\epm.
\ee
\end{proof}
\begin{remark}
When $k=1$ above, we have $P_k^+=Id,\ P_k^-=0$, then we can easily check that the decomposition reduces to $-\Delta_x=D_x^2$.
\end{remark}
Hence, we also have a matrix form for the four Rarita-Schwinger type operators. Further, one can obtain a Stokes' Theorem for $\bs{D}_1$ as following.
\begin{theorem}[Stokes' Theorem]\label{ST1}
Given functions $\bs{f},\bs{g}\in C^{\infty}(\Rm,\Hk)$, we have
\begin{align*}
&\int_{\partial\Omega}\int_{\Sm}\overline{\bs{f}(x,u)}^T\bs{d\sigma_x}\bs{g}(x,u)dS(u)\\
&=\int_{\Omega}\int_{\Sm}(\overline{\bs{D}_1\bs{f}(x,u)})^T\bs{g}(x,u)dS(u)dx+\int_{\Omega}\int_{\Sm}\overline{\bs{f}(x,u)}^T(\bs{D_1}\bs{g(x,u)})dS(u)dx,
\end{align*}
where $\bs{d\sigma_x}=\bpm 1 &0\\0&-\sqrt{\frac{m+2k}{m+2k-4}}\epm n(x)d\sigma(x)$, n(x) is the outward unit normal vector on $\partial\Omega$, $d\sigma(x)$ is the area element on $\partial\Omega$ and $dS(u)$ is the area element on the unit sphere.
\end{theorem}
One can also have a Borel-Pompeiu formula as follows.
\begin{theorem}[Borel-Pompeiu formula]
Let $\bs{f}\in C^{\infty}(\Omega,\Hk)$, then we have
\baa
f(y,v)
=&\int_{\partial\Omega}\int_{\Sm}\overline{\mathscr{E}_k(x-y,u,v)}^T\bs{d\sigma_x}\bs{f}(x,u)dS(u)\\
&-\int_{\Omega}\int_{\Sm}\overline{\mathscr{E}_k(x-y,u,v)}^T(\bs{D_1}\bs{f}(x,u))dS(u)dx.
\end{align*}
\end{theorem}
\section{Conclusion}
The main contribution of this paper is to introduce a new matrix form for bosonic Laplacians, which allows us to factorize these rather complicated second order differential operators as a multiplication of a first order differential operator and its transpose. Further, this result gives rise to a new matrix form for the Rarita-Schwinger type operators as well, and some properties for the matrix type bosonic Laplacians and Rarita-Schwinger operators are provided in more concise forms. This leads to a possibility to a further study on the Rarita-Schwinger type operators and bosonic Laplacians with matrix forms.

\end{document}